\documentclass[conference]{IEEEtran}

\makeatletter
\def\ps@headings{%
\def\@oddhead{\mbox{}\scriptsize\rightmark \hfil \thepage}%
\def\@evenhead{\scriptsize\thepage \hfil \leftmark\mbox{}}%
\def\@oddfoot{}%
\def\@evenfoot{}}
\makeatother

\usepackage{amssymb,color,amsmath}
\usepackage{cite}
\usepackage{subfigure}
\usepackage{paralist,color,clock}
\usepackage{graphicx}
\usepackage{algorithmic,algorithm2e}

\newcommand{\N}{\mathcal{N}}

\newtheorem{theorem}{\textbf{Theorem}}
\newtheorem{lemma}[theorem]{\textbf{Lemma}}
\newtheorem{defn}[theorem]{\textbf{Definition}}

\newcommand{\nix}[1]{}

\begin{document}
\title{\LARGE{Data Dissemination And Collection Algorithms For Collaborative Sensor Networks Using Dynamic Cluster Heads}}

\author{
\authorblockN{ Salah A. Aly$^{1,2}$~~~~~~~~~~~~~ Mohamed Salim$^{1}$ \\\\ \goodbreak
$^{1}$Center of Research Excellence in Hajj and Umrah, Umm Al-Qura University, Makkah, KSA\\
$^{2}$College of Computer and Information Systems, Umm Al-Qura University, Makkah, KSA\\
Emails: salahaly@uqu.edu.sa,~~~~~~~~
msalim@crowdsensing.net
 }
}
  \maketitle

\begin{abstract}
We develop novel  data dissemination and collection algorithms for Wireless Sensor Networks (WSNs) in which we consider $n$  sensor  nodes  distributed randomly in a certain field to measure a physical phenomena. Such sensors have limited energy, shortage coverage range, bandwidth and memory constraints. We desire to disseminate nodes' data throughout the network such that a base station will be able to collect the sensed data by querying a small number of nodes. We propose two data dissemination and collection algorithms (DCA's) to solve this problem. Data dissemination is achieved through dynamical selection of some nodes.  The selected nodes will be changed after a time slot $t$ and may be repeated after a period $T$.\footnote{Thanks to HajjCoRE, Center of Research Excellence in Hajj~ and~ Umrah at~ UQU, and NPSTI at KACST in KSA, ~ agencies for funding this work.}

\end{abstract}
\section{Introduction}\label{sec:intro}

Wireless Sensor Networks (WSNs) are expanding rapidly due to various applications and ease of development. However, WSNs encounter several challenges  to be deployed efficiently in a given environment. Such challenges are limited source energy, limited transmission bandwidth, shortage coverage range, data dissemination, data persistence, redundancy of defective nodes and data security. A typical wireless sensor network (WSN) can be used in many applications such as monitoring a physical phenomena from the surrounding environment like temperature, gases, humidity, volcanoes and tornados. Also, it can be used in animal tracking, forest fire detection and in military applications such as  detection of enemy intrusion.

Many techniques are used in data dissemination \cite{Imran10}, \cite{Filipovic07} and cluster head election \cite{Levente09}, \cite{Younis04}, \cite{Liu07c}. Fountain codes and random walks have been used to disseminate data from $\kappa$  sources to a set of storage nodes $\tau$, see \cite{Filipovic08},  \cite{Filipovic09}.
LEACH algorithm \cite{Handy02} is the most popular clustering algorithm. Several cluster head selection algorithms are based on LEACH architecture. The main drawback of the mentioned techniques is the requirement that  positions of all sensors must be known. Our algorithms don't use Fountain codes or random walks and independent on sensors positions.

We  consider a model for large-scale wireless sensor networks with $n$ identical sensing nodes distributed randomly and uniformly in a certain field. The nodes do not know the locations of the neighboring nodes as required in \cite{Dimakis06} and they don't maintain routing tables. In this work, we propose two algorithms for data dissemination and data collection in wireless sensor networks. The first algorithm is Pre-known Head selection data Dissemination and Collection Algorithm (PHDCA). The second  algorithm is Random Head selection data Dissemination and Collection Algorithm (RHDCA). We aim to develop an efficient method to randomly distribute and collect information from $n$ sensors by querying $10\%-20\%$ of nodes for retrieving information about all network nodes with a high probability.

\begin{figure}
  \includegraphics[width=9cm,height=5cm]{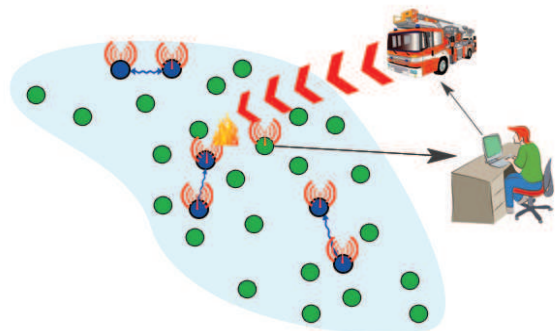}
  \caption{A model for WSNs with $n$ nodes distributed randomly and uniformly, among them are k cluster head nodes with a blue color.}\label{fig:wsn1}
\end{figure}

This work is organized as follows. In Section~\ref{sec:relatedwork} we present a background and  short survey of the related work. In Section~\ref{sec:model} we introduce the network model. In Sections~\ref{sec:alg1} and ~\ref{sec:analysis}, we describe  the DCA's algorithms and demonstrate some analysis for the DCA's algorithms.   In Section~\ref{sec:simulation} we present simulation studies of the proposed algorithms, and the paper is concluded in Section~\ref{sec:conclusion}.

\section{Network Model}\label{sec:model}
In this section we present the network model. Consider a set of $n$ identical sensing nodes distributed randomly in a field $F$ of dimensions $A=L\times W$, where $L$ and $W$ are the length and the width of $F$, respectively. We assume that each node has at least one neighboring node, meaning that with  probability $P=1$ there are no isolated nodes.

\begin{defn}[Cluster head] The cluster head node (HN) is an arbitrary node among all network nodes $\N$ which exchanges its neighbors data with the other neighboring cluster head nodes.
\end{defn}

\begin{defn}[Node degree] The node degree $d_{s_i}$ is the number of neighboring nodes to the node $s_i$ within its coverage range. The average mean degree of all nodes in $\N$ is given by
\begin{eqnarray}
\mu=\frac{1}{n}\sum_{i=1}^n d_{s_i}.
\end{eqnarray}
\end{defn}
The total period $(T)$ is the period after which the sensed data has been disseminated in the network $\N$ and is divided into $\epsilon$ equal time slots:
 \begin{equation}\label{eq1}
    T=\epsilon \times t,
 \end{equation}
for some integer number $\epsilon$, The algorithms performance and simulation results confirm our theoretic bounds.
The head nodes consume more energy than other nodes due to excess transmissions needed for data dissemination and data collection. So, such nodes are dynamically selected to apply fairness in energy consumption on all nodes. Also, the dynamical selection improves the performance of data dissemination in the network.
The head nodes will be changed every time slot $t$. The number of head nodes in the network is $k$ (where $k/n \cong \%10$). The selection of $T$ depends on the intended application (i.e. $T$ is small for high data rate applications and large for low data rate applications).

\noindent\textbf{Assumptions:}
  \begin{enumerate}
    \item Let $S=\{s_1,......,s_n\}$ be a set of $n$ identical sensing nodes distributed randomly in a field $F$ of dimensions $A=L\times W$, where $L$ and $W$ are the length and the width of $F$, respectively.
    \item Let $H=\{h_1,......,h_k\}$ be a set of $k$ head nodes selected from the $n$ sensing nodes to disseminate the data in the network and they will be changed at each time slot $t$.
    \item Let $T$ be the period after which the sensed data has been disseminated in the network and it is divided into $\epsilon$ equal slots $t=\{t_1,..., t_{\epsilon}\}$.
    \item The nodes use flooding to know their neighbors, as each node will send a message containing its $ID_{s_i}$ to all neighboring nodes. Each node receives an incoming $ID_{s_i}$ from any node $s_i$ will consider the node of the incoming $ID_{s_i}$ as its neighbor.
    \item Each node in the network generates a packet $P_{s_i}$ as follows:
    \begin{equation}\label{eq2}
        P_{s_i}=(ID_{s_i},x_{s_i},flag),
    \end{equation}
     where, $ID_{s_i}$ is the $ID$ of node $s_i$, $x_{s_i}$ is the sensed data of node $s_i$ and $flag$ is a variable set to $0$ in flooding process or to $1$ otherwise.
    \item Each node has a radio range coverage $r_i$. The node $s_i$ will be considered as a neighbor of $s_j$ if and only if $d_{s_i,s_j}\leq r_j$, where $d_{s_i,s_j}$ is the distance between nodes $s_i$ and  $s_j$.
    \item  Initially, let the number of nodes $n$ is known. Practically, the number of nodes in the network varies due to node energy depletion, failure nodes and added redundance nodes. Hence, it is important to estimate the number of nodes at each period $T$. The base station will consider the number of retrieved nodes when querying 10\% of nodes as the total number of nodes $n$. The estimated number of nodes will be sent to the first survived head node in the network (i.e. The first survived node from H).
  \end{enumerate}

\section{DCA'S ALGORITHMS}\label{sec:alg1}
In this section, we will demonstrate the DCA's algorithms.
\subsection{PHDCA ALGORITHM}
In PHDCA algorithm we dynamically select the $k$ cluster head nodes that disseminate the data in the network according to a pre-known manner. The algorithm can be classified into four phases as follows:

\begin{itemize}
  \item \textbf{Initialization phase:} In this phase, the head nodes are initially selected from $ ID_{s_i} =1 : 0.1n $ at the first time slot $t_1$.
  \item \textbf{ Flooding phase:}  In this phase, each sensor will broadcast a message containing its $ID_{s_i}$ to be able to discover its neighbors to store them in its data base. If any node receives any incoming $ID_{s_i}$, it will consider the node of the incoming $ID_{s_i}$ as its neighbor. Also, the broadcasting message containing a flag equal zero to indicate the flooding phase.
  \item \textbf{Sensing and data dissemination phase:} In this phase, such sensor reads a new data, it will send this data to some of its neighboring nodes. The neighboring head nodes will disseminate the data in the network by exchanging their neighbors data among them as shown in Fig.~\ref{fig:wsn1}. The head nodes will be changed at each time slot and repeated each period $T$.
  \item \textbf{Data collection phase:} In this phase, the base station can query small number of any nodes to retrieve the data sensed by the $n$ sensing nodes and make an estimation for $n$ to send it to the first survived node.
\end{itemize}

\vspace{.5 cm}
\begin{algorithm}[t!]

\KwIn{A sensor network with $S=\{s_1,\ldots,s_n\}$ source nodes,  $n$ source packets $x_{s_i},\ldots,x_{s_n}$.}%
\KwOut{storage buffers $y_1,y_2,\ldots,y_n$ for all sensors $S$.} %
$t=1$; //initiate the value that represents the number of time slot in the period $T$.
\\\ForEach{node $u=1:n$}%
    {
     \If{$u\leq0.1n$}%
                {
                 $u$ is a head node;%
                }%

    }%
    \ForEach{node $u=1:n$}%
    {
     Generate a packet containing $ID_{u}$ , $flag=0$ and broadcast this message to its set of neighbors;%
     \\$P_{u}=(ID_{u},x_{u},flag)$;%
    }%

\While{still remains surviving nodes}%
    {
    \ForEach{node $u=1:n$}%
    {
     \If{u sensed new data}%
                {
                 $u$ will send this data to some of its neighbors randomly;%
                }%

    }%
    \ForEach{ head node h= 1:k}%
    {
     $h$ and its neighboring head nodes exchange their neighbors data with each others ;
    }%
     \If{t expired}%
                {
                 Generate new $k$ head nodes as follows:
                \\ $t$ ++ ;
                \\ \ForEach{node $u=1:n$}%
                   {
                      \If{$ 0.1 n (t-1)< u \leq0.1 n t$}%
                     {
                        $u$ is a head node;
                      }%

                   }%

                  \If{$t==\epsilon$}%
                     {
                        $t$=0;
                       \\$n=n_{received}$; //updates $n$ by the received estimated node number from base station.
                      }%
                }%
        }%

\mbox{}\\

\caption{ PHDCA algorithm } \label{alg:PHDCA}
\end{algorithm}

\subsection{RHDCA ALGORITHM}
In PHDCA algorithm, we assumed that the selection of head nodes is pre-known at each time slot $t$ and the head nodes are repeated each period $T$. The disadvantage of this algorithm is the topology dependence. So, its performance depends on the distribution of head nodes which depends on network topology. We extended PHDCA to obtain RHDCA that randomly selects $k$ head nodes at each time slot $t$. The performance of RHDCA is topology independent due to randomly selection of head nodes.
The difference between the two algorithms is the sensing and data dissemination phase as follows:
\\ \textbf{Sensing and data dissemination phase:} In this phase $k$ head nodes are selected randomly at each time slot $t$. The $k$ head nodes may be not repeated each period $T$. Also, each sensor reads a new data, it will send this data to some of its neighboring nodes. The neighboring head nodes will disseminate the data in the network by exchanging their neighbors data among them.
\goodbreak

\section{ANALYSIS}\label{sec:analysis}
In this section we  analyze the proposed DCA's algorithms.
\begin{lemma}
The probability that a set $M$ of sensors has at least one cluster head node is given by
\begin{equation}\label{eq4}
\Pr(M\cap H)=1- \prod \limits_{i=1}^m (1-\frac{k}{n-i+1}),
\end{equation}
where,  $m=|M|$ is the number of nodes in $M$.
\end{lemma}

\begin{proof}
Number of ways in which the $m$ nodes can be drawn from the total number of nodes $n$ is
\begin{eqnarray} \binom {n}{m}=C_m^n=\frac{n!}{m!  (m-n)!}.
 \end{eqnarray}
 Number of ways so that no head nodes exist in the set $M$ is $\binom {n-k}{m}$. So, the probability that  the set $M$ has no cluster head nodes is $\frac{\binom {n-k}{m}}{\binom {n}{m}}$. Hence, the probability that the set $M$ has at least one head node is
 \begin{eqnarray}
    1-\frac{\binom {n-k}{m}}{\binom {n}{m}}&=& 1- \prod \limits_{i=1}^m (1-\frac{k}{n-i+1}) \nonumber
 \end{eqnarray}

\end{proof}


\begin{lemma}The probability that a set $M$ of sensors has a set $Z$ of cluster head nodes is given by
\begin{equation}\label{p3}
 \Pr( Z )=\frac{\binom {n-k}{m-z} \binom {k}{z}}{\binom {n}{m}},
\end{equation}
where, $z=|Z|$ is the number of nodes in $Z$.
\end{lemma}

\begin{proof}
Number of ways in which the $m$ nodes can be drawn from the $n$ sensing nodes is $\binom {n}{m}$. From the Fundamental Counting Theorem, the total number of ways in which $z$ head nodes and $m-z$ non head nodes can be drawn from the $n$ sensing nodes is $ \binom {n-k}{m-z} \binom {k}{z}$. So, The probability that a set of $n$ sensor has $z$ head nodes is $ \frac{ \binom {n-k}{m-z} \binom {k}{z}}{\binom {n}{m}}$.
\end{proof}

\begin{defn}[Head energy consumption ($E_{h}$)] is the energy consumption at all nodes due to data dissemination in the network $\N$ when all nodes have the same coverage range and packet size.
\end{defn}
\begin{lemma}
Let  $\beta$ be the probability that a node $s_i$ has a set $Z$ of  neighboring head nodes. From Equation~(\ref{p3}), $\beta$ can be given by

 \begin{equation}\label{A1}
    \beta= \frac{\binom {n-k}{d_{s_i}-z_{s_i}}  \binom {k}{z_{s_i}}}{\binom {n}{d_{s_i}}},
 \end{equation}
where, $z_{s_i}$ is the number of neighboring head nodes to node $s_i$ and $d_{s_i}$ is the degree of node $s_i$ when the set $M$ represents the neighboring nodes of the node $s_i$.
\\ The total energy consumption $E_h$ is given by
\begin{equation}\label{eq6}
E_{h}=\frac{\epsilon k}{n}(n \mu p_r + p_t\sum\limits_{i=1}^n \beta \alpha_i z_{s_i} ),
\end{equation}
where, $\alpha_i$ is the number of transmissions between the node $s_i$ and its neighbors and $p_t$, $p_r$ are the transmitted and received energy costs due to sending one packet.
\end{lemma}

\begin{proof}
Let $\sigma$ be the received energy cost of nodes $n$ due to data dissemination, so
\begin{eqnarray}\label{rcost}
    \sigma&=&\frac{k}{n} \sum\limits_{i=1}^n d_{s_i} p_r \nonumber \\
    &=&k \mu p_r,
\end{eqnarray}
where, $\frac{k}{n}$ is the probability that a node $s_i$ is a head node. Let $\zeta$ be the transmitted energy cost of nodes $n$ due to data dissemination, so
\begin{eqnarray}\label{tcost}
    \zeta &=&\frac{k}{n} \sum_{i=1}^n \beta \alpha_i z_{s_i} p_t \nonumber \\
        &=&\frac{k p_t}{n} \sum\limits_{i=1}^n \beta \alpha_i z_{s_i}.
\end{eqnarray}
Hence, the total energy consumption due to data dissemination at $\epsilon$ time slots is given by
\begin{eqnarray}\label{rcost}
     E_{h}&=& \epsilon  (\sigma+\zeta)\nonumber \\
    &=& \frac{\epsilon k}{n}(n \mu p_r +  p_t\sum\limits_{i=1}^n \beta \alpha_i z_{s_i}).
\end{eqnarray}
\end{proof}

\begin{lemma}
The total energy consumption at the sensing nodes due to sending the sensed data to their neighbors is given by
\begin{equation}\label{eq7}
E_{s}=n(p_t+\mu p_r),
\end{equation}
where all nodes have the same coverage range and packet size.
\end{lemma}

\begin{proof}
The energy consumption at nodes $n$ due to sending its sensed data is $n p_t$. The energy consumption at nodes $n$ due to all received packets is $\sum\limits_{i=1}^n p_r\times d_{s_i}$. Hence, assuming that each node updates its data one time at each period $T$, the energy consumption at the $n$ sensing nodes is
 \begin{eqnarray}\label{eq7}
E_{s}&=&\sum\limits_{i=1}^n \Large(\large p_t+d_{s_i} p_r\large)\large \nonumber\\
&=&n(p_t+ \mu p_r).
\end{eqnarray}
\end{proof}

\begin{lemma}
The optimum number of head nodes that gives minimum energy consumption is given by
\begin{equation}\label{eq}
    k_{opt}=\sqrt{\mid\frac{n E_h}{\epsilon \lambda \sum\limits_{i=1}^n \alpha_{i} Z_{s_i}}\mid}
\end{equation}
where, $\lambda=\frac{d}{dk}\sum\limits_{i=1}^n \beta $.
\end{lemma}
\begin{proof}
The optimum number of head nodes that gives minimum energy consumption can be driven by the differentiation of Equation (~\ref{eq6}) as follows:
\begin{equation}\label{eq9}
    \frac{d}{dk_{opt}}\left( E_h\right)=0.
\end{equation}
\end{proof}

\section{Simulation and Performance Evaluations}\label{sec:simulation}
In this section we present some simulation results to illustrate the performance of the proposed algorithms.

\begin{defn} Decoding Ratio ($\eta$) is the ratio between the number of queried nodes $\hat{n}$ and the total number of sources $n$.
\begin{equation}\label{10}
   \eta=\frac{ \hat{n} }{n}
\end{equation}
\end{defn}

\begin{defn} Successful Decoding Probability ($P_s$) is the probability that the $n$ source packets are all recovered from the $\hat{n}$ queried nodes.
\end{defn}

\begin{figure}
    \centering

    \subfigure[PHDCA]
    {
        \includegraphics[height=4cm,width=8cm]{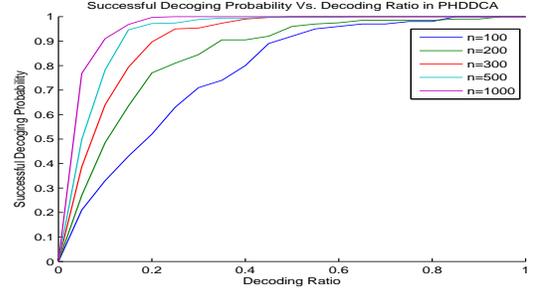}
        \label{fig:various_n}
    }
    \subfigure[RHDCA]
    {
        \includegraphics[height=4cm,width=8cm]{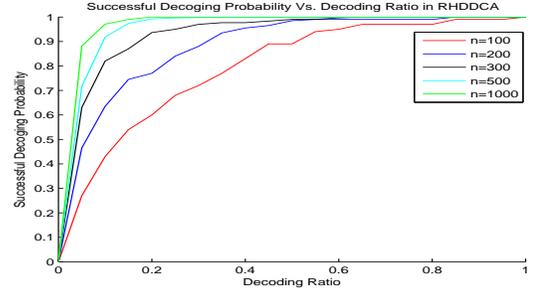}
        \label{fig:R_prop_n}
    }
    \caption{The relation between the successful decoding probability and the decoding ratio for $n$=100, $n$=200, $n$=300, $n$=500, $n$=1000 when $A$=100*100, $\epsilon$=10, buffer size=40 and $r$=5.}
    \label{fig:probability}
\end{figure}
\begin{figure}
    \centering

    \subfigure[PHDCA]
    {
        \includegraphics[height=4cm,width=9cm]{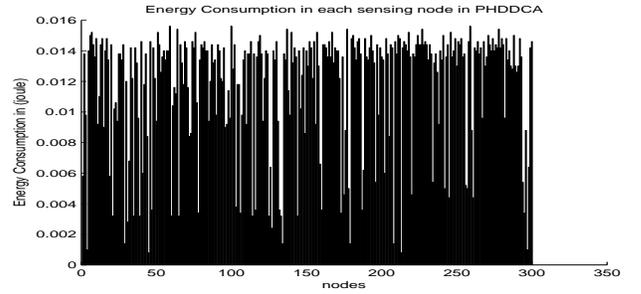}
        \label{fig:energy_cons}
    }
    \subfigure[RHDCA]
    {
        \includegraphics[height=4cm,width=9cm]{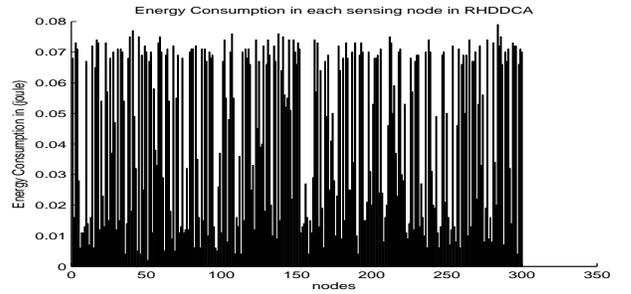}
        \label{fig:R_energy}
    }
    \caption{The energy consumption at each sensing node in network $N$ when $A$=100*100, n=300, $\epsilon$=10, buffer size=40, packet size=2Kbits, node energy=5J and $r$=5.}
    \label{fig:energy}
\end{figure}

Fig.~\ref{fig:probability} shows the relation between the successful decoding probability and the decoding ratio for different values of sensing nodes $n$ in PHDCA and RHDCA algorithms.

Fig.~\ref{fig:various_n} and Fig.~\ref{fig:R_prop_n} show that increasing the number of network nodes $n$ and fixing the covering radius $r$ of all nodes will result in an improvement in the successful decoding probability as well. We can notice that as the number of nodes increasing, the number of queried sensors can be decreased to recover the data with a reasonable successful probability. Particularly, for $n > 500$, we see that querying up to 10\% will reveal about 85\% of network data in PHDCA and about 92\% of network data in RHDCA.

Fig.~\ref{fig:energy} shows the amount of energy consumption at each node after the dissemination of data in the network $\N$ in PHDCA and RHDCA algorithms. From this figure we can notice that the energy consumption in PHDCA algorithm is better than the obtained result in RHDCA algorithm. We assumed that the energy consumption at the sensing node due to sensing  data is neglected and each sensor node is assumed to be of initial
battery charge $5$ Joule. We calculated the energy consumption according to \cite{Meghanathan10}, they assumed that the energy consumption at a sensor node $s_i$ due to transmitting  one packet is given by
  \begin{equation}\label{5}
        p_t=(50 *10^{-9}  + 100 *10^{-12}*r_{s_i} ^2)* \psi_{s_i}.
  \end{equation}
and the energy consumption at a sensor node $s_i$ due to receiving  one packet is given by

  \begin{equation}\label{6}
        p_r=50 *10^{-9}* \psi_{s_i},
  \end{equation}
where $\psi_{s_i}$ is the packet size of node $s_i$.
\begin{defn} Death Rate (DR) is the ratio between the number of dead nodes $\bar{n}$ and the total number of sensing nodes $n$.
\end{defn}
\begin{equation}\label{D_R}
    DR=\frac{\bar{n}}{n}.
\end{equation}

\begin{figure}
    \centering

    \subfigure[PHDCA]
    {
        \includegraphics[height=4cm,width=8cm]{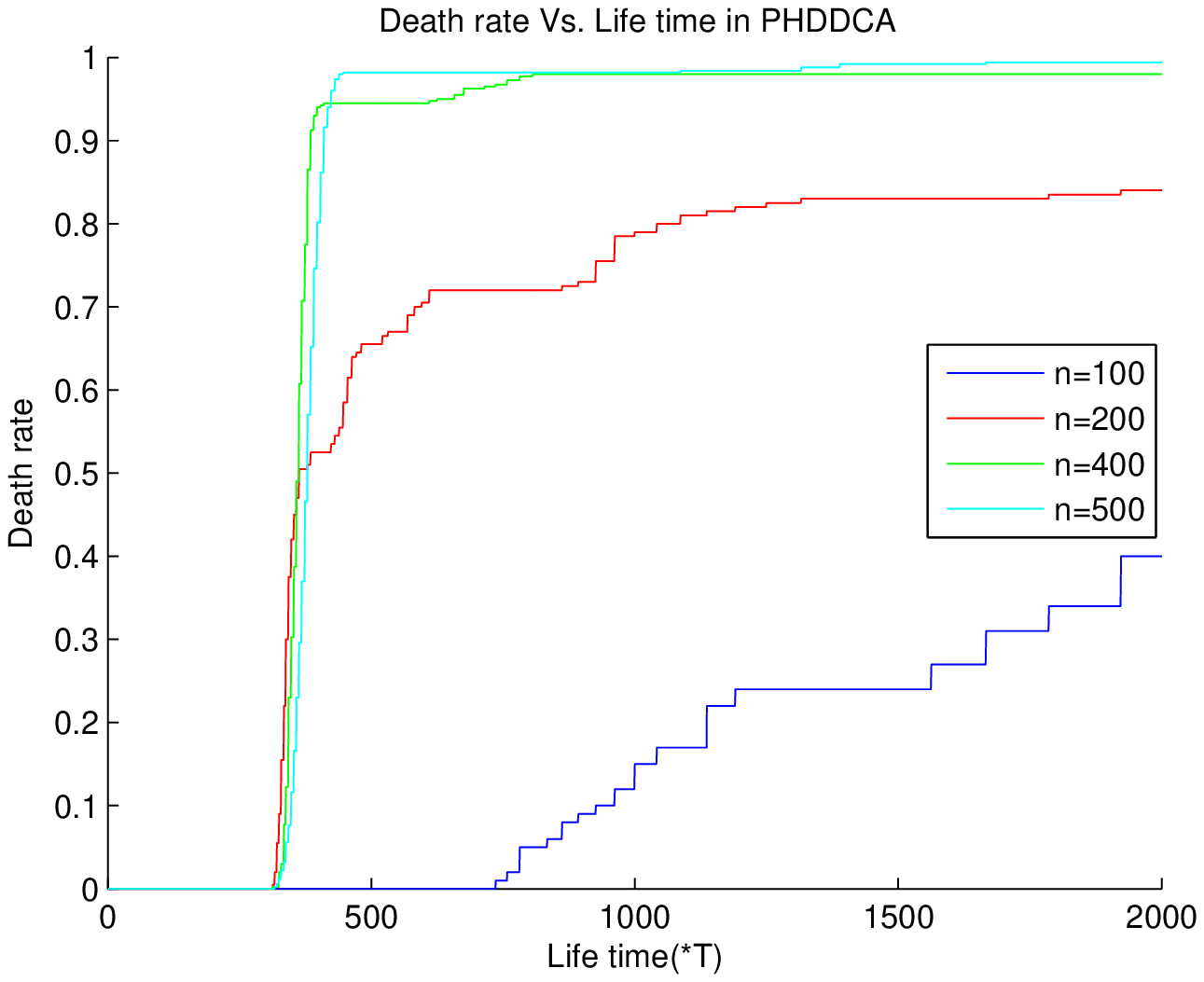}
        \label{fig:PHDCA_death}
    }
    \subfigure[RHDCA]
    {
        \includegraphics[width=8cm,height=4cm]{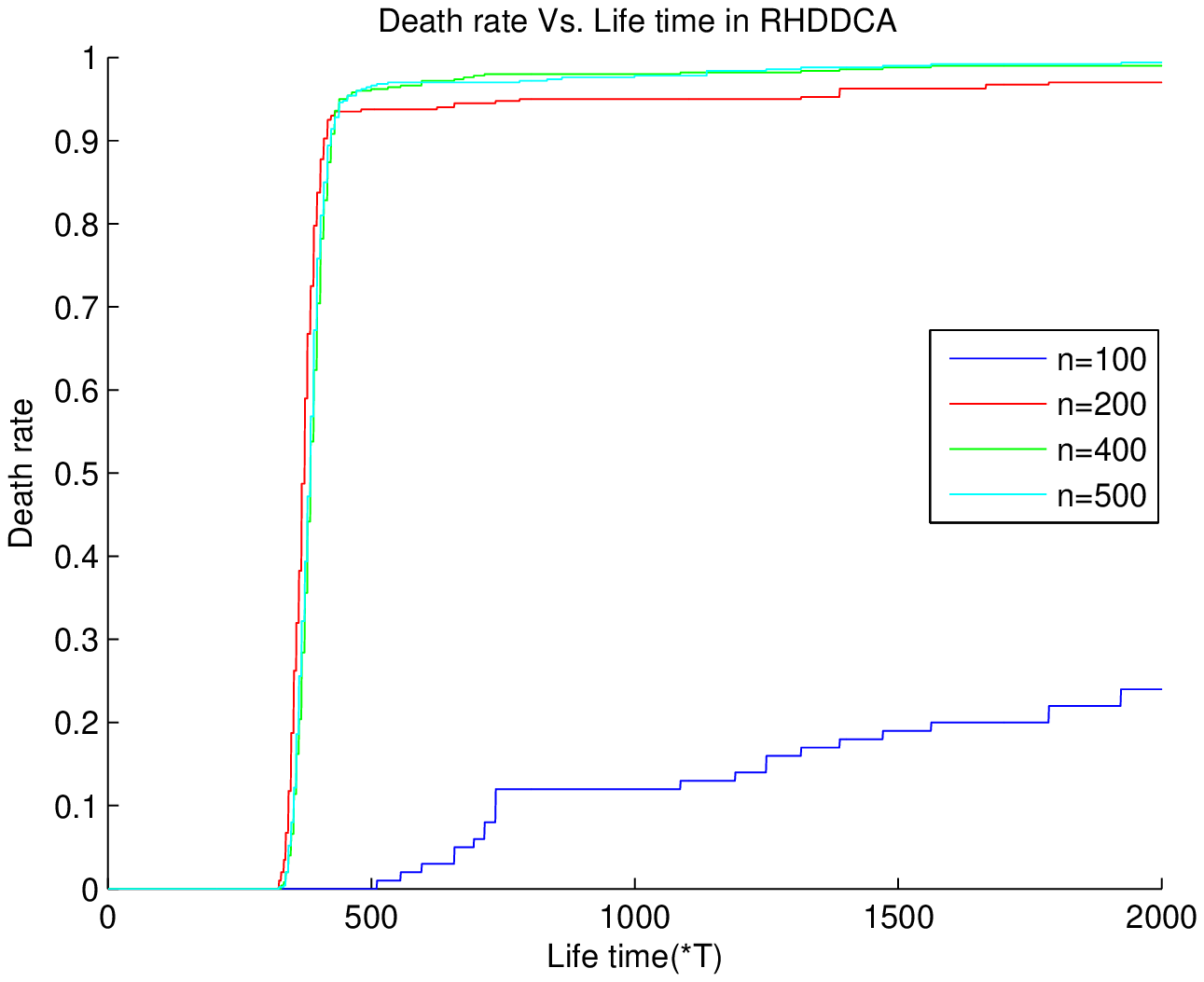}
        \label{fig:RHDCA_death}
    }
    \caption{The relation between the death rate and number of nodes $n$ in network $N$ when $A$=100*100, $\epsilon$=10, buffer size=40 and $r$=5.}
    \label{fig:death}
\end{figure}

Fig.~\ref{fig:death} illustrates the relation between the death rate and the total number of sensing nodes. Increasing the number of network nodes $n$ and fixing the field area will result in an increase in the death rate as well.

Fig.~\ref{fig:performance} shows the relation between the performance of data collection and the elapsed time in DCA's. As the elapsed time increases, more nodes disappear from the network $N$ (i.e. $DR$ increases) due to energy depletion. Hence, the data dissemination and data collection performances will be negatively affected by the disappeared nodes $\bar{n}$. From Fig.~\ref{fig:PHDCA_perf_along_time} and~\ref{fig:RHDCA_perf_along_time} we can deduce that the performance of data collection of PHDCA along time is better than RHDCA.

\begin{figure}
    \centering

    \subfigure[PHDCA]
    {
        \includegraphics[height=4cm,width=8cm]{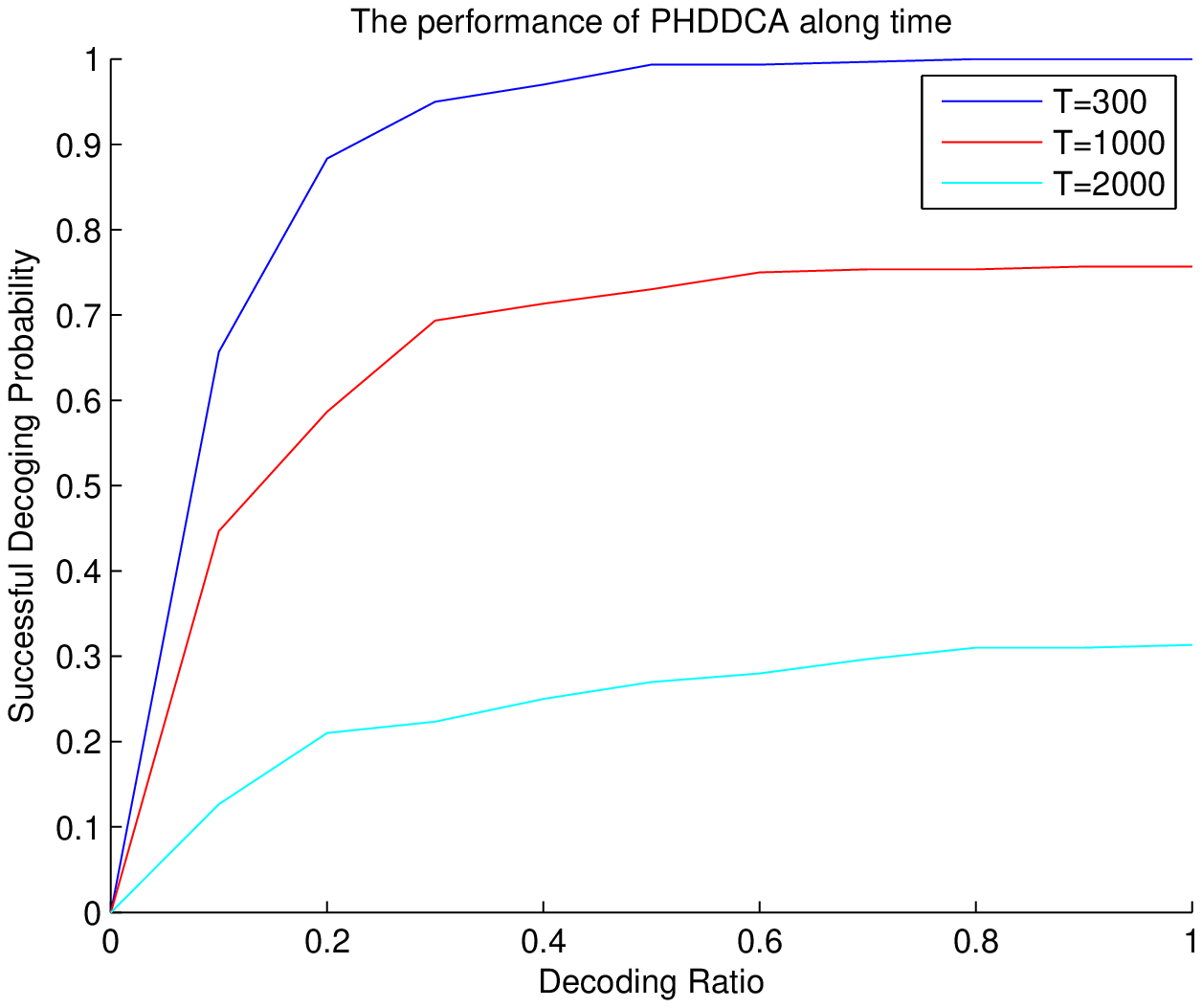}
        \label{fig:PHDCA_perf_along_time}
    }
    \subfigure[RHDCA]
    {
        \includegraphics[height=4cm,width=8cm]{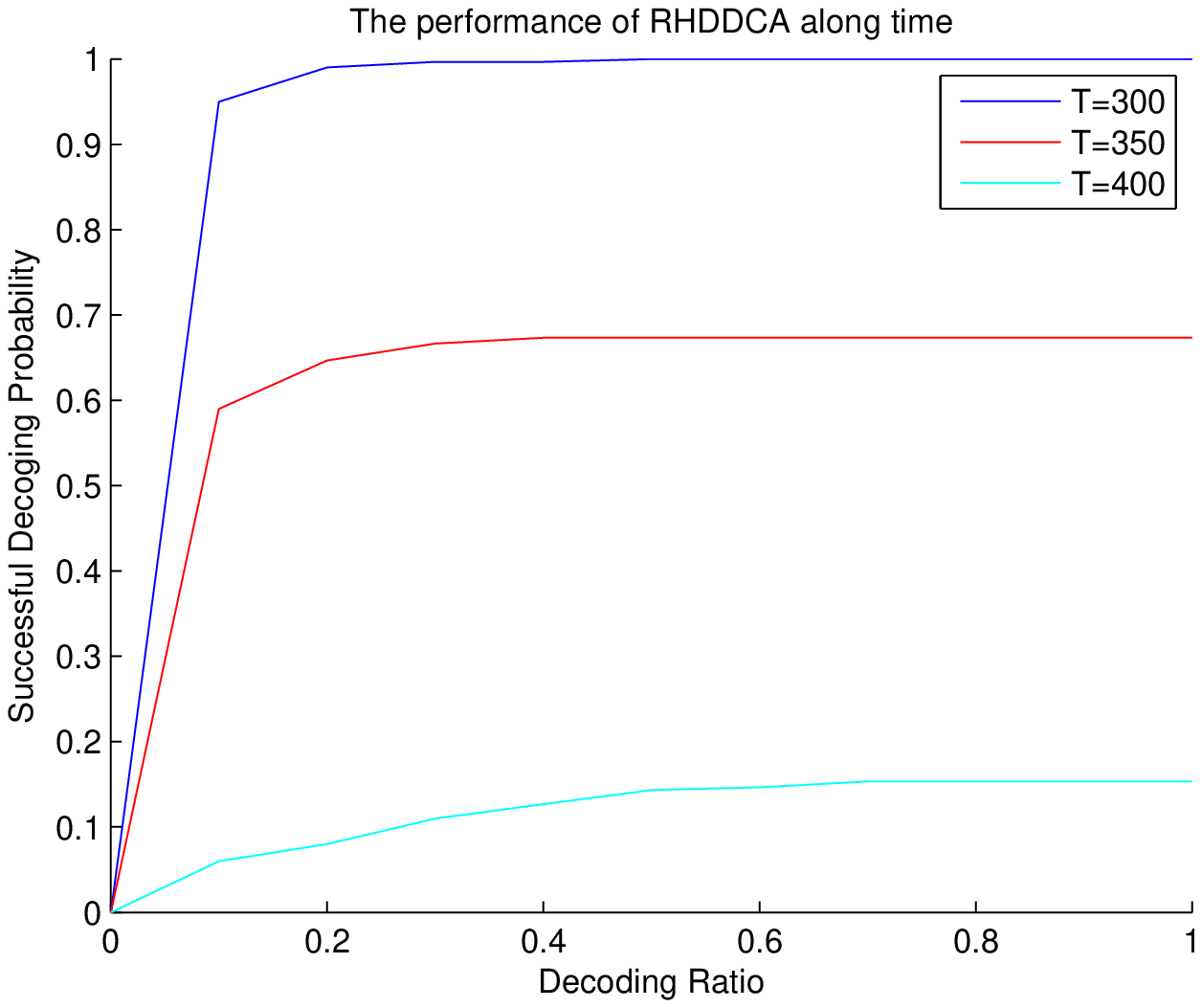}
        \label{fig:RHDCA_perf_along_time}
    }
    \caption{The performance of PHDCA and RHDCA algorithms along time when $A=100\times100$, $n=300$, $\epsilon$=10, $buffer size=40$ and $r$=5.}
    \label{fig:performance}
\end{figure}
\section{Discussions}\label{sec:relatedwork}
In this section we will indicate the related work to our work.

\begin{compactitem}
  \item The authors in \cite{Aly11} proposed a distributed data collection algorithm to store and forward information obtained by wireless sensor networks. They used $n-k$ storage nodes to collect the sensed data from the network, where $k$ is the sensor nodes, $n$ is the total number of nodes and $(n-k)/n$ is 20\%.
 \item The authors in \cite{Aly09}, \cite{Kong10}, \cite{Aly08} suggested two distributed storage algorithms for large-scale wireless sensor networks. Such node chooses randomly one of its neighbors to send its data to another neighbor. Such packet can travel to a certain number of hops according to its time to live counter. The receiver node will decide with a random probability if it will accept the incoming message or not.
 \item The authors in \cite{Dimakis05} used a decentralized implementation of Fountain codes that uses geographic routing and every node has to know its location.
  \item The authors in  \cite{Kamra06} increased data persistence in wireless sensor networks using a novel technique called growth codes, i.e. increasing the amount of information that can be recovered at the sink nodes.
  \item Authors in \cite{Heinzelman02} used a centralized controller to select CHs. By using a central control algorithm to form the clusters, it can produce better clusters by dispersing the CH nodes throughout the network.
  \item The authors in \cite{Bandyopadhyay05} proposed a distributed, randomized clustering algorithm to organize the sensors in a wireless sensor network into clusters.
\end{compactitem}
\section{Conclusion and future work}\label{sec:conclusion}
In this paper, we presented two algorithms for data dissemination and collection in wireless sensor networks. Given $n$ sensing nodes with limited buffers, we demonstrated schemes to disseminate sensed data throughout the network with less computational overhead. The proposed algorithms did not assume any pro-known of routing tables or nodes' locations. In addition, the time factor $T$ can be selected to be suitable for the intended applications and minimizing energy consumption. Our future work will develop accurate practical algorithms to optimize energy consumptions in the sensor network.

\section*{Acknowledgments}

This research is supported by the Center of Research Excellence in Hajj and Omrah (HajjCoRE) at Umm Al-Qura University in Makkah, KSA, under a project number P1127, entitled "Crowd-Sensing System: Crowd Estimation, Safety, and Management", 2012.  Also, this research is supported by the NPSTI grant from KACST, KSA.

\bibliographystyle{plain}

\end{document}